\documentclass [13pt]{article}
\usepackage{amsthm}
\usepackage{amsmath}
\usepackage{amscd,amsfonts,amsmath,amssymb,amsthm}
\usepackage{geometry}
\usepackage{tikz}
\usetikzlibrary{positioning}
\usepackage{pgf,tikz}
\usetikzlibrary{arrows, automata}

\newtheorem{theorem}{Theorem}[section]

\newtheorem{lemma}[theorem]{Lemma}

\usepackage{doc}

\usepackage{url}

\usepackage{graphicx}
\usepackage{epstopdf}
\title{The Shapley Value of Cyclic Digraph Games}
\author{Krishna Khatri \footnote{Department of Mathematics, Piedmont College, Demorest, GA. I am grateful to my advisor, Dr. Douglas Torrance, and math department of Piedmont College.}}

\begin{document}
\maketitle

\abstract
{
In this paper the Shapley value of digraph (directed graph) games are considered. Digraph games are transferable utility (TU) games with limited cooperation among players, where players are represented by nodes. A restrictive relation between two adjacent players is established by a directed line segment. Directed path, connecting the initial player with the terminal player, form the coalition among players. A dominance relation is established between players and this relation determines whether or not a player wants to cooperate. To cooperate, we assume that a player joins a coalition where he/she is not dominated by any other players. The Shapley value \cite{khmelnitskaya2016shapley} is defined as the average of marginal contribution vectors corresponding to all permutations that do not violate the subordination of players. The Shapley value for cyclic digraph games is calculated and analyzed. For a given family of characteristic functions, a quick way to calculate Shapley values is formulated. 
}

\section{Introduction}
\label{introduction}
Game theory is the mathematical theory that studies the conflict and cooperation 
between rational decision makers. Game theory helps to analyze decision making between two or more individuals who influence one another's welfare  \cite{myerson2013game}.
Cooperative game theory deals with coalitions and allocations, and considers group of players willing to allocate the joint benefits derived from their cooperation  \cite{gonzalez2010introductory}.

When the players in a game form a coalition to work together, it is essential to identify the correct way to distribute the profit among themselves. If some of the players in the coalition are unsatisfied with the proposed allocation, then they are free to leave the coalition. In stable coalitions there are fewer incentives to leave the coalition. The Shapley value provides a unique way to divide a payoff among players in such a way as to satisfy various fairness criteria. Distributing payoff to all players according to their Shapley value helps to create a stable coalition. Myerson considers the cooperation between players in an undirected graph, where each player has an equal chance to move away from a coalition by breaking the path between them  \cite{myerson1977graphs}. Such games assume fair and equal gain through cooperation. 

This paper is motivated by the paper ``The Shapley value for directed graph games" of Anna Khmelnitskaya, Ozer Selcuk and Dolf Talman. They introduce the Shapley value for digraph games and look for its stability \cite{khmelnitskaya2016shapley}.

As the structure of this paper, digraph games and the Shapley value are defined in section $2$ and the following theorem is proved in section $3$. 
\begin{theorem} \label{thm:neat}
Consider $f:\mathbb{Z}_{\geq 0}\rightarrow \mathbb{R}$ with $f(0)=0$. Suppose $N=\{1,...,n\}$ and define $v_f : 2^N \rightarrow \mathbb{R}$ by $v_f (S)= f(|S|)$. Let $\Gamma$ be the directed cycle $(1,2,...,n,1)$. Then the Shapley value of the digraph game $(v_f,\Gamma)$ is $$Sh(v_f,\Gamma)=\left( \smash[b]{\! \underbrace{\dfrac{f(n)}{n},\dfrac{f(n)}{n},\dfrac{f(n)}{n},\cdots, \dfrac{f(n)}{n}\,}_\text{n times}}\right).$$

\end{theorem}
\vspace{5mm}
Finally in section $4$, Shapley values of various directed cycle games are calculated. 

\section{Preliminaries}
A \textit{cooperative transferable utility (TU) game} is a pair $(N,v)$, where $N=\{1,...,n\}$ is a finite set of players with $n \geq 2$ and $v:2^N \rightarrow \mathbb{R}$. We interpret $v(S)$ as the payoff that the coalition $S\subseteq N$ can generate. By convention, the payoff of an empty coalition is zero \textit{i.e.} $v(\O)=0$. We refer to the set of all TU-games with fixed set of players $N$ as $\mathbb{G}^N$. For simplicity we use $v$ to refer to $(N,v)$. For any player $i\in N$, player $i$'s minimum payoff which he can guarantee to himself without joining any coalition is $v(\{i\})$ . 

A \textit{digraph} is a tuple $\Gamma=(N,\tau)$ where $N$ is a finite set of players and $\tau$ is a set of directed edges. A \textit{subgraph H} of $\Gamma$ is a digraph whose sets of players and directed edges are subsets of $N$ and $\tau$, respectively. The restriction of a digraph $\Gamma$ to a coalition $S$ is denoted by $\Gamma|_S$. A \textit{directed path} is a sequence $(k_1,k_2,...,k_m)$ of players such that the directed edge $(k_1,k_{i+1})$ is in $\tau$ for all $i$. A \textit{directed cycle} of players is a directed path with $k_m=k_1$. A player $j$ is \textit{successor} of player $i$ if there exists a directed path from $i\in N$ to $j\in N$ in $\Gamma$. For $i\in N$, $S^\Gamma (i)$  denotes the set of successors of $i$ in $\Gamma$ and $\bar{S}^\Gamma (i)=S^\Gamma (i)\cup \{i\}$.
For digraph $\Gamma$ and $S\subseteq N$, player $i\in S$ \textit{dominates} player $j\in S$ in $\Gamma|_S$ if $j\in S^{\Gamma|_S}(i)$ and $i\notin S^{\Gamma|_S}(j) $. When a player does not have any  predecessors, then he or she is \textit{undominated}. No player is dominated on directed cycle.

A \textit{digraph game} is a pair $(v, \Gamma)$ of a TU-game $v\in \mathbb{G}^N$ and a digraph $\Gamma$. A permutation  $\pi \in \Pi$ is \textit{consistent} in $\Gamma $ if it preserves the subordination of players determined by $\Gamma, \textit{i.e.}$, $j \succ _{\Gamma|_{\bar{P}_\pi(i)}} i$ implies $\pi(j) > \pi(i)$ .   

 \textit{The marginal contribution} of player $i\in N$ to the coalitions in a game $v\in \mathbb{G}^N$ is given by $m_i ^v(S)=v(S\cup {i})-v(S)$. For any $i\in N$ and permutation $\pi: N\rightarrow N$, $\pi(i)$ is the position of player $i$ in $\pi$. A player $i$ is a \textit{predecessor} of player $j$ in $\Gamma$ if there exists a directed path from $i$ to $j$ in $\Gamma$. The set of predecessors of $i$ in $\pi$ is denoted $P_\pi (i)$ and $\bar{P}_\pi (i)=P_\pi (i) \cup \{i\}$. For any $i\in N$ on a TU game, the marginal contribution vector $ \bar{m}^v (\pi)\in \mathbb{R}^N$ is $\bar{m}_i^v (\pi)= m_i^v (P_\pi (i))=v(\bar{P}_\pi (i))-v(P_\pi (i))$. 
 
 \textit{The Shapley value} of a TU game is $$Sh(v,\Gamma)=\sum_{\pi \in \Pi} \dfrac{\bar{m}^v (\pi)}{|\Pi^\Gamma|},$$ where $\Pi^\Gamma$ is the set of all permutations on $N$ which are consistent with $\Gamma$. In \textit{the grand coalition} $N$, players divide $v(S)$ among themselves. The outcome of this division depends on the power structure in the grand coalition. The Shapley value provides a fair way to distribute $v(N)$ among themselves \cite{gonzalez2010introductory}.

\section{Results}
\begin{lemma}
Suppose $\Gamma$ is the cyclic digraph $(1, 2, ..., n, 1)$. The only permutations which are consistent with $\Gamma$ are $\pi_k = (-1 + k, -2 + k, ..., 1 + k, k)$, where addition is modulo $k$, for each $k = 1, ..., n$.
\end{lemma}
\begin{proof}
\begin{figure}[!ht]
\centering
\begin{tikzpicture}[shorten >=1pt, auto, node distance=3cm, scale=0.5, ultra thick]
    \tikzstyle{node_style} = [circle,draw=black,font=\sffamily\Large\bfseries]
    \tikzstyle{edge_style} = [draw=black, line width=2, ultra thick]
    \node[node_style,label=$1+k_n$]  (v1) at (0,-1) {};
    \node[node_style,label=$n$] (v2) at (0,3) {};
    \node[node_style,label=$1$] (v3) at (3,5) {};
    \node[node_style,label=$2$] (v4) at (6,3) {};
    \node[node_style,label=$-1+k_n$] (v5) at (6,-1) {};
    \node[node_style,label=$k_n$] (v6) at (3,-3) {};
    \draw[dashed,->]  (v1) edge (v2);
    \draw[->]  (v2) edge (v3);
    \draw [->] (v3) edge (v4);
    \draw[dashed,->]  (v4) edge (v5);
    \draw[->]  (v5) edge (v6);
    \draw[->]  (v6) edge (v1);
    
    \end{tikzpicture}
    \caption{$\Gamma$}
\end{figure}
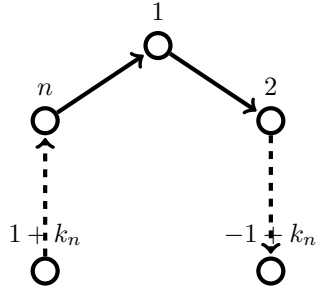
\begin{figure}[!ht]
\centering
\begin{tikzpicture}[shorten >=1pt, auto, node distance=3cm, scale=0.5, ultra thick]
    \tikzstyle{node_style} = [circle,draw=black,font=\sffamily\Large\bfseries]
    \tikzstyle{edge_style} = [draw=black, line width=2, ultra thick]
    \node[node_style,label=$1+k_n$]  (v1) at (0,-1) {};
    \node[node_style,label=$n$] (v2) at (0,3) {};
    \node[node_style,label=$1$] (v3) at (3,5) {};
    \node[node_style,label=$2$] (v4) at (6,3) {};
    \node[node_style,label=$-1+k_n$] (v5) at (6,-1) {};
    
    \draw[dashed,->]  (v1) edge (v2);
    \draw[->]  (v2) edge (v3);
    \draw [->] (v3) edge (v4);
    \draw[dashed,->]  (v4) edge (v5);

    \end{tikzpicture}
    \caption{$\Gamma|_{\bar{P}_{\pi}(k_{n-1})}$}
\end{figure}

Suppose $\pi=(k_1,k_2,...,k_{n-1},k_n)$ is consistent, so $\bar{P}_\pi (k_{n-1})=\{k_1,k_2,...,k_{n-1}\}$. By removing $k_n$ from $\Gamma$ in Figure $1$, $\Gamma|_{\bar{P}_{\pi}(k_{n-1})}$ is the directed path $(1+k_n, 2+k_n,...,-1+k_n)$ as shown in Figure $2$. So $1+k_n \succ 2+k_n \succ ... \succ -1+k_n$. Thus by the definition of consistency, $\pi(1+k_n) > \pi (2+k_n) > ... > \pi(-1+k_n)$. It follows that $k_{n-1}=1+k_n, k_{n-2}=2+k_n,...,-1+k_n=k_1$. Hence, $\pi = \pi_{k_{n}}$

\end{proof}

\begin{proof}[Proof of Theorem \ref{thm:neat}]
Let $\Gamma$ be the directed cycle with  a characteristic function $v_f (S)=f(|S|)$, where $S\subseteq N$ is a coalition of players.
 By Lemma 3.1, the number of players $n$ is equal to the number of permutations that are consistent with $\Gamma$.  We know that for any player $i$, the $i^{th}$ component of the Shapley value is  \begin{align*}
     Sh(v_f, \Gamma)_i &=\dfrac{1}{|\Pi^\Gamma|}\sum_{\pi \in \Pi^\Gamma }\bar{m}^v (\pi)\\
     &=\dfrac{1}{n}\sum_{\pi \in \Pi^\Gamma }\bar{m}^v(\pi)\\
     &=\dfrac{1}{n}\sum_{\pi \in \Pi^\Gamma }(v(\bar{P}_\pi (i))-v(P_\pi(i)))\\
     &=\dfrac{1}{n}\sum_{\pi \in \Pi^\Gamma }(v({P}_\pi (i))\cup v(i))-v(P_\pi (i))).
 \end{align*} We also know $P_\pi (i)=\{j\in N|\pi(j)<\pi(i)\}$. For each $k\in \{1,...,n\}$, there exists exactly one permutation $\pi$ for which $|\bar{P}_{\pi} (i)|=k$. Since $v_f (S)=f(|S|)$, the marginal contribution of player $i\in N$ is $\sum_{j=1^n} (f(j) - f(j-1)).$. This is equivalent to $\sum_{j=1}^n (f(j)-f(j-1))$. So,
\begin{align*}
Sh(v_f,\Gamma)_i & =\dfrac{\sum_{j=1}^n \left(f(j)-f(j-1)\right)}{n} \\
&= \dfrac{\left(f(n)-f(n-1)\right)+ \left(f(n-1)-f(n-2)\right)+ ...+ \left(f(2-1)-f(1-1)\right)}{n}&&
\end{align*}
 This is a telescoping sum $i.e.$ each term in the numerator cancels except the initial and final terms. Thus, Sh($V_f,\Gamma)_i= \dfrac{f(n) -f(0)}{n}=\dfrac{f(n)}{n}$.

\end{proof}

\section{Examples}
For any coalition $S\subseteq N$ and $k\in \mathbb{N}$, consider the characteristic function defined as  $v_k(S) = |S|^k$ on digraph $\Gamma$.

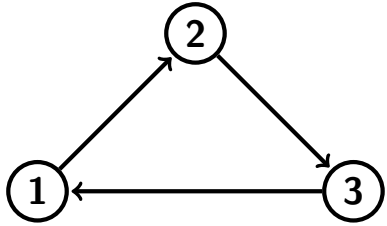
\begin{figure}[!ht]\
\centering
\begin{tikzpicture}[shorten >=1pt, auto, node distance=3cm,scale=0.7, ultra thick]
    \tikzstyle{node_style} = [circle,draw=black,font=\sffamily\Large\bfseries]
    \tikzstyle{edge_style} = [draw=black, line width=2, ultra thick]
    \node[node_style]  (v1) at (3,3) {2};
    \node[node_style] (v2) at (6,0) {3};
    \node[node_style] (v6) at (0,0) {1};
    \draw[->]  (v1) edge (v2);
    \draw[->]  (v2) edge (v6);
    
    \draw[->]  (v6) edge (v1);
    
    \end{tikzpicture}
    \caption{Digraph $\Gamma$}
\end{figure}

Consider a cyclic digraph game $(v,\Gamma)$ with three different players as shown in Figure $3$.  The set of all permutations that is consistent with $\Gamma$ is $\{(1,3,2),(2,1,3),(3,2,1)\}$. For $k=0,  Sh(v_0,\Gamma)=\left(\dfrac{1}{3},\dfrac{1}{3},\dfrac{1}{3}\right).$ For $k=1, Sh(v_1,\Gamma)=(1,1,1).$ For $k=2, Sh(v_2,\Gamma)=(3,3,3).$ For $k=3, Sh(v_3,\Gamma)=(9,9,9). $ For $k=4, Sh(v_4,\Gamma)=(27,27,27)$, and so on.\\

\begin{figure}[!ht]
\centering
\begin{tikzpicture}[shorten >=1pt, scale=0.7,auto, node distance=3cm, ultra thick]
    \tikzstyle{node_style} = [circle,draw=black,font=\sffamily\Large\bfseries]
    \tikzstyle{edge_style} = [draw=black, line width=2, ultra thick]
    \node[node_style]  (v1) at (3,3) {2};
    \node[node_style] (v2) at (6,0) {3};
    \node[node_style] (v3) at (3,-3) {4};
    \node[node_style] (v6) at (0,0) {1};
    \draw[->]  (v1) edge (v2);
    \draw[->]  (v2) edge (v3);
    \draw[->]  (v3) edge (v6);
    \draw[->]  (v6) edge (v1);
    
    \end{tikzpicture}
    \caption{Digraph $\Gamma'$}
\end{figure}

Again, consider a cyclic digraph game $(v,\Gamma')$ with four different players as shown in Figure $4$.  The set of all permutations that is consistent with $\Gamma'$ is $\{(1,4,3,2),(2,1,4,3),(3,2,1,4),(4,3,2,1)\}$.  For $k=0, Sh(v_0,\Gamma')=\left(\dfrac{1}{4},\dfrac{1}{4},\dfrac{1}{4},\dfrac{1}{4}\right).$ For $k=1, Sh(v_1,\Gamma')=(1,1,1,1).$ For $k=2, Sh(v_2,\Gamma')=(4,4,4,4).$ For $k=3, Sh(v_3,\Gamma')=(16,16,16,16). $ For $k=4, Sh(v_4,\Gamma')=(64,64,64,64),$ and so on.
\begin{figure}[!ht]
\centering
\begin{tikzpicture}[shorten >=1pt, auto, node distance=3cm, scale=0.5, ultra thick]
    \tikzstyle{node_style} = [circle,draw=black,font=\sffamily\Large\bfseries]
    \tikzstyle{edge_style} = [draw=black, line width=2, ultra thick]
    \node[node_style]  (v1) at (4,3) {2};
    \node[node_style] (v2) at (8,0) {3};
    \node[node_style] (v3) at (6,-4) {4};
    \node[node_style] (v4) at (2,-4) {5};
    \node[node_style] (v5) at (0,0) {1};
    \draw[->]  (v1) edge (v2);
    \draw[->]  (v2) edge (v3);
    \draw[->]  (v3) edge (v4);
    \draw[->]  (v4) edge (v5);
    \draw[->]  (v5) edge (v1);
    
    \end{tikzpicture}
    \caption{Digraph $\Gamma''$}
\end{figure}

As shown in Figure $5$, consider a cyclic digraph game $(v,\Gamma'')$ with five different players.  The set of all permutations that is consistent with $\Gamma''$ is $\{(1,5,4,3,2),(2,1,5,4,3),(3,2,1,5,4),(4,3,2,1,5),\\(5,4,3,2,1)\}$.  For $k=0, Sh(v_0,\Gamma'')=\left(\dfrac{1}{5},\dfrac{1}{5},\dfrac{1}{5},\dfrac{1}{5},\dfrac{1}{5}\right).$ For $k=1, Sh(v_1,\Gamma'')=(1,1,1,1,1).$ For $k=2, Sh(v_2,\Gamma'')=(5,5,5,5,5).$ For $k=3, Sh(v_3,\Gamma'')=(25,25,25,25,25). $ For $k=4, Sh(v_4,\Gamma'')=(125,125,125,125,125), $ and so on.
\newpage
\bibliographystyle{unsrt}
\bibliography{mybib}

\end{document}